\documentclass{revtex4}

\usepackage{amsmath}
\usepackage{amssymb}
\usepackage{amsthm}
\usepackage{graphicx}
\usepackage{color}
\definecolor{darkred}{rgb}{0.8,0,0}

\newtheorem{defn}{Definition}
\newtheorem{thm}{Theorem}
\newtheorem{prop}[thm]{Proposition}
\newtheorem{lem}[thm]{Lemma}

\DeclareMathOperator{\rk}{rk}
\DeclareMathOperator{\brk}{\underline{rk}}
\DeclareMathOperator{\slto}{\xrightarrow{SLOCC}}
\DeclareMathOperator{\mean}{\mathbb{E}}
\DeclareMathOperator{\supp}{supp}

\newcommand{\cut}[2]{#1:#2}

\begin{document}

\title{Asymptotic entanglement transformation between W and GHZ states}
\author{P\'eter Vrana}
\affiliation{Institute for Theoretical Physics, ETH Z\"urich, Wolfgang-Pauli-Strasse 27, CH-8093 Z\"urich, Switzerland}
\affiliation{Department of Geometry, Budapest University of Technology and Economics, Egry J\'ozsef u. 1., 1111 Budapest, Hungary}
\author{Matthias Christandl}
\affiliation{Institute for Theoretical Physics, ETH Z\"urich, Wolfgang-Pauli-Strasse 27, CH-8093 Z\"urich, Switzerland}
\affiliation{Department of Mathematical Sciences, University of Copenhagen, Universitetsparken 5, 2100 Copenhagen, Denmark}

\date{\today}

\begin{abstract}
We investigate entanglement transformations with stochastic local operations and classical communication (SLOCC) in an asymptotic setting using the concepts of degeneration and border rank of tensors from algebraic complexity theory. Results well-known in that field imply that GHZ states can be transformed into W states at rate $1$ for any number of parties. As a generalization, we find that the asymptotic conversion rate from GHZ states to Dicke states is bounded as the number of subsystems increase and the number of excitations is fixed. By generalizing constructions of Coppersmith and Winograd and by using monotones introduced by Strassen we also compute the conversion rate from W to GHZ states.
\end{abstract}

\maketitle

\section{Introduction}

Understanding entanglement in multipartite quantum states is one of the major goals in quantum information theory. While entanglement in bipartite pure states is well understood both in the local operations and classical communication (LOCC) and stochastic local operations and classical communication (SLOCC) paradigms at the single copy level and asymptotically, much less is known for three or more parties. Among the few results are the complete classification of pure three and four qubit states under SLOCC equivalence \cite{threeqbit,fourqbit}. Already for three qubits one finds two incomparable genuinely tripartite-entangled classes, indicating the complexity of the problem.

Recently a connection between algebraic complexity theory and the study of SLOCC transformations in an asymptotic setting has been discovered \cite{tripartite}, opening the possibility to transfer ideas from one field to the other. More precisely, it has been observed that finding the rate at which triples of EPR pairs shared among three parties can be extracted from GHZ states is the same as finding the exponent of matrix multiplication, commonly denoted by $\omega$, a problem that has been studied for over 40 years by mathematicians \cite{Strassen1,Bini,Sch,CW}. This number is the infimum of real numbers $\tau$ such that $n\times n$ matrices can be multiplied together using $O(n^\tau)$ arithmetic operations. Later, bounds on the tensor rank of multiple copies of the W state have been found \cite{rkWn}, as well as its generalization to more than three parties and other symmetric states \cite{catalysis}.

In this paper we further explore this connection and show that border rank and degeneration, two important concepts in algebraic complexity, can be used to prove nontrivial bounds on asymptotic conversion rates. These techniques are well-known in algebraic complexity theory, but seem not to have been applied so far in the present context.

The structure of the paper is as follows. In section \ref{sec:aSLOCC} we provide the definition of asymptotic SLOCC conversion rates between two states and give some basic properties, including its connection to tensor rank and asymptotic rank. In section \ref{sec:border} we introduce the concepts of degeneration and border rank into the study of asymptotic SLOCC transformations and illustrate their usefulness by computing the conversion rate from GHZ to (generalized) W states for any number of parties. Section \ref{sec:symm} extends the latter result to an upper bound on the conversion rate from GHZ states to certain families of symmetric states, uniformly in the number of subsystems. In section \ref{sec:WtoGHZ} we compute the conversion rate from W to GHZ states. The construction follows an idea of Coppersmith and Winograd which was used to prove an upper bound on the exponent of matrix multiplication. Optimality is shown using the monotones introduced by Strassen.

During the preparation of this manuscript we have learned about independent related work by Yu, Guo and Duan \cite{Wrate1} where they give a proof of theorem \ref{thm:GHZtoW}.

\section{Asymptotic SLOCC transformations}\label{sec:aSLOCC}

Given a pair of $k$-partite pure states $\psi\in\mathcal{H}_1\otimes\cdots\otimes\mathcal{H}_k$ and $\varphi\in\mathcal{K}_1\otimes\cdots\otimes\mathcal{K}_k$ we denote by $\psi\slto\varphi$ the fact that there exist linear transformations $A_i:\mathcal{H}_i\to\mathcal{K}_i$ such that $\varphi=(A_1\otimes\cdots\otimes A_k)\psi$. This is the well-known mathematical condition for the ability of $k$ parties to transform the state $\psi$ into $\varphi$ with nonzero probability, each being able to control one of the subsystems, while their actions are allowed to be coordinated via classical communication \cite{threeqbit}. It is clear that this definition completely ignores the normalization of the states (in fact, only the vector space structure is used and not the norm). For simplicity we therefore prefer to work with unnormalized states.

We are interested in transformations of multiple copies where the number of initial and final copies may be different. The relevant quantity is
\begin{equation}
\omega_n(\psi,\varphi):=\frac{1}{n}\inf\{m\in\mathbb{N}|\psi^{\otimes m}\slto\varphi^{\otimes n}\}
\end{equation}
where the infimum of the empty set is considered to be $\infty$.
It is easy to see that $(n_1+n_2)\omega_{n_1+n_2}(\psi,\varphi)\le n_1\omega_{n_1}(\psi,\varphi)+n_2\omega_{n_2}(\psi,\varphi)$, which implies that the limit $\omega(\psi,\varphi):=\lim_{n\to\infty}\omega_n(\psi,\varphi)$ exists and is equal to $\inf\omega_n(\psi,\varphi)$.

It was observed by Chitambar et al. \cite{tripartite} that when $k=3$ and we let $\psi=GHZ$ and $\varphi$ be the triple of EPR pairs, one shared between each pair of subsystems, then $\omega(\psi,\varphi)$ is precisely the exponent of matrix multiplication, the smallest real number $\tau$ such that for any $\varepsilon>0$ two $n\times n$ matrices can be multiplied using $O(n^{\tau+\varepsilon})$ arithmetic operations. This exponent is traditionally denoted by $\omega$ -- our notation is chosen so that it emphasizes this connection to algebraic complexity.
SLOCC transformation rates have also been investigated in ref. \cite{Wrate1}, the relation between our quantity and theirs is
\begin{equation}
R(\psi,\varphi)=\frac{1}{\omega(\psi,\varphi)}
\end{equation}

For $S\subseteq[k]$ let $\rk_S\psi$ denote the Schmidt rank of $\psi$, considered as a bipartite state on the subsystems $S$ and $\bar{S}=[k]\setminus S$. The following facts are simple consequences of the definition, and therefore we omit the proofs.
\begin{prop}\label{prop:basicn}
For any choice of the appearing states and bipartitions the followings hold:
\begin{enumerate}
\item $\displaystyle\omega_n(\psi,\varphi)\ge\max_{S\subseteq[k]}\frac{\log\rk_S\varphi}{\log\rk_S\psi}$
\item $\omega_n(\varphi,\varphi)=1$
\item $\omega_n(\varphi_1,\varphi_3)\le\omega_{n\omega_n(\varphi_2,\varphi_3)}(\varphi_1,\varphi_2)\omega_n(\varphi_2,\varphi_3)$
\item $\omega_n(\psi,\varphi_1\otimes\varphi_2)\le\omega_n(\psi,\varphi_1)+\omega_n(\psi,\varphi_2)$
\item $\omega_n(\psi_1\otimes\psi_2,\varphi)\le\max\{\frac{\lfloor\alpha n\rfloor}{n}\omega_{\lfloor\alpha n\rfloor}(\psi_1,\varphi),\frac{n-\lfloor\alpha n\rfloor}{n}\omega_{n-\lfloor\alpha n\rfloor}(\psi_2,\varphi)\}$ for any $0\le\alpha\le 1$.
\item $\omega_n(\psi_1\otimes\psi_2,\varphi_1\otimes\varphi_2)\le\max\{\omega_n(\psi_1,\varphi_1),\omega_n(\psi_2,\varphi_2)\}$
\item $\omega_n(\psi^{\otimes b},\varphi^{\otimes a})=\frac{1}{n}\left\lceil\frac{n a}{b}\omega_{n a}(\psi,\varphi)\right\rceil$
\end{enumerate}
\end{prop}
By taking the limit $n\to\infty$ we get the following useful asymptotic forms:
\begin{prop}[Basic properties of asymptotic SLOCC conversion rate]\label{prop:basica}
For any choice of the appearing states and any subset $S\subseteq[k]$ the followings hold:
\begin{enumerate}
\item $\displaystyle\omega(\psi,\varphi)\ge\max_{S\subseteq[k]}\frac{\log\rk_S\varphi}{\log\rk_S\psi}$
\item $\omega(\varphi,\varphi)=1$
\item $\omega(\varphi_1,\varphi_3)\le\omega(\varphi_1,\varphi_2)\omega(\varphi_2,\varphi_3)$
\item $\omega(\psi,\varphi_1\otimes\varphi_2)\le\omega(\psi,\varphi_1)+\omega(\psi,\varphi_2)$
\item $\omega(\psi_1\otimes\psi_2,\varphi)\le\frac{1}{\frac{1}{\omega(\psi_1,\varphi)}+\frac{1}{\omega(\psi_2,\varphi)}}$
\item $\omega(\psi_1\otimes\psi_2,\varphi_1\otimes\varphi_2)\le\max\{\omega(\psi_1,\varphi_1),\omega(\psi_2,\varphi_2)\}$
\item $\omega(\psi^{\otimes b},\varphi^{\otimes a})=\frac{a}{b}\omega(\psi,\varphi)$
\end{enumerate}
\end{prop}

We will make use of the $a$-level generalization of GHZ states:
\begin{equation}
GHZ_a=\sum_{i=1}^a|ii\ldots i\rangle
\end{equation}
where the number of parties should always be clear from the context. We omit the subscript when $a=2$.

In algebraic complexity these are called unit tensors and have a special role due to their connection to tensor rank \cite{Gastinel,Strassen2,tripartite}. Tensor rank itself can be seen as a generalization of matrix rank (or Schmidt rank), and is defined as the minimum number of product states spanning a subspace containing a given tensor. We denote tensor rank by $\rk$.
\begin{prop}
Let $k\in\mathbb{N}$. For any $k$-partite state $\psi$ the followings hold:
\begin{enumerate}
\item $\omega_n(GHZ_a,\psi)=\frac{1}{n}\lceil\log_a\rk\psi^{\otimes n}\rceil$
\item $\omega(GHZ_a,\psi)=\frac{1}{\log a}\lim_{n\to\infty}\frac{1}{n}\log\rk\psi^{\otimes n}$
\item $\omega_n(GHZ_a,GHZ_b)=\frac{1}{n}\left\lceil\frac{n\log b}{\log a}\right\rceil$
\item $\omega(GHZ_a,GHZ_b)=\frac{\log b}{\log a}$
\end{enumerate}
\end{prop}
\begin{proof}
\item It is known that $GHZ_m\slto\psi^{\otimes n}$ iff $m\ge \rk \psi^{\otimes n}$ \cite{tripartite}. In the sequence $GHZ_a, {GHZ_a}^{\otimes 2}=GHZ_{a^2}, {GHZ_a}^{\otimes 3}=GHZ_{a^3}, \ldots$ of unit tensors of rank $a^i$ the first one with rank not less than $\rk \psi^{\otimes n}$ has index $\lceil\log_a \rk \psi^{\otimes n}\rceil$. Now let $n\to\infty$ and use $\rk (GHZ_b)^{\otimes n}=b^n$ to get the remaining three equalities.
\end{proof}
The quantity $2^{\omega(GHZ,\psi)}$ is also known as the asymptotic rank of $\psi$.

Since the rank of a tensor is always finite, one can extract any state from GHZ states at a positive rate ($\omega(GHZ,\psi)<\infty$). In the other direction, a state clearly needs to be globally entangled if we are to distill GHZ states out of many copies of it, i.e. it cannot be biseparable across any bipartite cut. For two and three subsystems it is easy to see that this condition is also sufficient \cite{Strassen3}. It turns out that this is also true for more parties, as the following lemma shows, but the argument is more complicated in this case. A similar result is proved for exact LOCC transformations in ref. \cite{Comparability}.
\begin{lem}
Let $\psi$ be a globally entangled $k$-partite state and let $1\le i,j\le k$ label two specified subsystems. If $EPR_{i,j}$ denotes the state which consists of an EPR-pair shared between subsystems $i$ and $j$, and the rest is separable, then $\omega(\psi,EPR_{i,j})<\infty$.
\end{lem}
\begin{proof}
We prove by induction on the number of parties. For $k=2$ we have nothing to prove. If $k>2$ then there is a subsystem $c$ with $i\neq c\neq j$. Our goal is to find a local operation acting at this subsystem resulting in a state $|0\rangle_c\otimes\psi'$ where $\psi'$ is a globally entangled state on the remaining $k-1$ subsystems. Such an operation is clearly characterized by a linear map $\mathcal{H}_c\to\mathbb{C}$. For a proper subset $S\subseteq[k]\setminus\{c\}$ let us introduce the set
\begin{equation}
A_S=\{P:\mathcal{H}_c\to\mathbb{C}|\text{$(P\otimes I\otimes\cdots\otimes I)\psi$ is not biseparable across the cut $\cut{S}{[k]\setminus (S\cup\{c\})}$}\}
\end{equation}
These are a Zariski open subsets of the vector space of linear forms on $\mathcal{H}_c$, an irreducible affine variety. Our goal is to show that none of them is empty, which implies that their intersection is also not empty.

To this end let us do a Schmidt decomposition with respect to the bipartition $\cut{\{c\}}{[k]\setminus\{c\}}$:
\begin{equation}
\psi=\sum_{p=1}^r|p\rangle_c\otimes\psi_p
\end{equation}
if for some $p$ the state $\psi_p$ is not $\cut{S}{[k]\setminus (S\cup\{c\})}$-biseparable then the map $|p'\rangle\mapsto\delta_{pp'}|0\rangle$ is in $A_S$. On the other hand, if for every $p$ we have biseparability $\psi_p=\varphi^{S}_p\otimes\varphi^{S\cup\{c\}}_p$ with respect to this cut then we can always find a $p_1$ and $p_2$ such that $\varphi^{S}_{p_1}$ is not a multiple of $\varphi^{S}_{p_2}$ and $\varphi^{S\cup\{c\}}_{p_1}$ is not a multiple of $\varphi^{S\cup\{c\}}_{p_2}$ -- otherwise $\psi$ would not be globally entangled. In this case the map $|p'\rangle\mapsto(\delta_{p_1p'}+\delta_{p_2p'})|0\rangle$ is in $A_S$.
\end{proof}

Since $EPR_{1,2}\otimes EPR_{1,3}\otimes\cdots\otimes EPR_{1,k}$ can be converted to a GHZ via SLOCC using teleportation, this implies that
\begin{equation}
\begin{split}
\omega(\psi,GHZ)
 & \le\omega(\psi,EPR_{1,2}\otimes EPR_{1,3}\otimes\cdots\otimes EPR_{1,k})  \\
 & \le\omega(\psi,EPR_{1,2})+\omega(\psi,EPR_{1,3})+\cdots+\omega(\psi,EPR_{1,k})<\infty
\end{split}
\end{equation}

\section{Degeneration and border rank}\label{sec:border}

It is a standard fact that a GHZ state cannot be transformed into a W state by SLOCC \cite{threeqbit}, forming distinct entanglement classes of three qubits, but a W state can be approximated to arbitrary precision with states in the GHZ orbit \cite{Wappr}. Even though $\rk W=3$, it can be approximated by rank 2 GHZ states, and we say that its border rank is 2 (notation: $\brk W=2$).

This phenomenon is known as degeneration in algebraic complexity, and is important in the study of the complexity of tensor powers \cite{Bini,Strassen3}. More generally, we say that $\psi$ degenerates to $\varphi$ iff $\varphi$ is in the orbit closure of $\psi$ under the action of SLOCC. We remark that the closure in the Zariski topology is the same as that in the Euclidean topology. The following alternative definition is more convenient for calculations:
\begin{defn}
Let $\psi\in\mathcal{H}_1\otimes\cdots\otimes\mathcal{H}_k$ and $\varphi\in\mathcal{K}_1\otimes\cdots\otimes\mathcal{K}_k$ be two pure states. We say that $\psi$ \emph{degenerates} to $\varphi$ if there exist linear transformations $A_i(\varepsilon):\mathcal{H}_i\to\mathcal{K}_i$ depending polynomially on $\varepsilon$ such that
\begin{equation}
(A_1(\varepsilon)\otimes\cdots\otimes A_k(\varepsilon))\psi=\varepsilon^d\varphi+O(\varepsilon^{d+1})
\end{equation}
for some $d\in\mathbb{N}$.
\end{defn}
It can be shown that over algebraically closed base fields this algebraic definition is equivalent to the analytic one described above, see e.g. \cite{Burgisser}.

Just as the rank of a state $\psi$ can be characterized as the smallest $a$ such that $GHZ_a\slto\psi$, the border rank $\brk\psi$ is the smallest $a$ such that $GHZ_a$ degenerates to $\psi$.

We illustrate the concept using the W state as an example. Consider the following equality:
\begin{equation}\label{eq:Wappr1}
\frac{(|0\rangle+\varepsilon|1\rangle)\otimes(|0\rangle+\varepsilon|1\rangle)\otimes(|0\rangle+\varepsilon|1\rangle)-|000\rangle}{\varepsilon}=(|100\rangle+|010\rangle+|001\rangle)+\varepsilon(|011\rangle+|101\rangle+|110\rangle)+\varepsilon^2|111\rangle
\end{equation}
For any $\varepsilon\neq 0$ the tensor on the right hand side has rank $2$ and in the limit $\varepsilon\to 0$ it becomes the W state, hence $\brk W=2$. Note that the limit $\varepsilon\to 0$ can be seen as the derivative of a polynomial at $\varepsilon=0$. We can also understand the situation in a geometric way: states contained in \emph{secants} to the set of separable states have rank at most 2, while states on a \emph{tangent} to the set of separable states have border rank at most 2. The idea works for higher derivatives as well. The largest degree appearing on the right hand side ($2$ in the example) plays a role later, and is called the error degree of the approximation \cite{Sch}.

In algebraic complexity it is a well-known result that, asymptotically, degeneration and restriction (i.e. SLOCC convertibility) are equivalent \cite{Bini, Sch, Strassen3}. In the language of asymptotic SLOCC transformations the statement translates to the following:
\begin{thm}[Bini, Sch\"onhage, Strassen]
Let $\psi$ and $\varphi$ be $k$-partite states and suppose that $\psi$ degenerates to $\varphi$. Then $\omega(\psi,\varphi)\le 1$.
\end{thm}
The most general proof \cite{Sch, Strassen3} uses some nontrivial algebraic facts, but works for arbitrary base fields. Here we present a simplified version of the argument from \cite{Bini}, which works over algebraically closed fields of characteristic $0$. As the base field $\mathbb{C}$ is the most important in quantum physics, this level of generality is more than enough for our purposes.
\begin{proof}
Suppose first that $\psi$ is globally entangled, i.e. not biseparable across any bipartite cut. Then $\omega(\psi,GHZ)<\infty$. By the assumption we can write
\begin{equation}
(A_1(\varepsilon)\otimes\cdots\otimes A_k(\varepsilon))\psi=\varepsilon^d\varphi+\varepsilon^{d+1}\varphi_1+\cdots+\varepsilon^{d+e}\varphi_{e}
\end{equation}
for some error degree $e\in\mathbb{N}$ and states $\varphi_1,\ldots,\varphi_e$. Now take the $n$th tensor power of both sides. The resulting equation has the form
\begin{equation}
(A_1(\varepsilon)^{\otimes n}\otimes\cdots\otimes A_k(\varepsilon)^{\otimes n})\psi^{\otimes n}=\varepsilon^{nd}\varphi^{\otimes n}+\varepsilon^{nd+1}(\ldots)+\cdots+\varepsilon^{nd+ne}\varphi_{e}^{\otimes n}
\end{equation}
We show that this implies $\psi^{\otimes n}\otimes GHZ_{ne+1}\slto\varphi^{\otimes n}$. To this end let
\begin{equation}
A'_m=\sum_{j=1}^{ne+1}A_m(e^{\frac{2\pi i}{ne+1}j})^{\otimes n}\otimes|0\rangle\langle j|
\end{equation}
for $m=1,\ldots,k$ and
\begin{equation}
GHZ'_{ne+1}=\frac{1}{ne+1}\sum_{j=1}^{ne+1}e^{-\frac{2\pi i}{ne+1}ndj}|j\ldots j\rangle
\end{equation}
Then clearly $GHZ_{ne+1}\slto GHZ'_{ne+1}$ and
\begin{multline}
(A'_1\otimes\cdots\otimes A'_k)(\psi^{\otimes n}\otimes GHZ'_{ne+1})  \\
  = \frac{1}{ne+1}\sum_{j=1}^{ne+1}e^{-\frac{2\pi i}{ne+1}ndj}\left(A_1(e^{\frac{2\pi i}{ne+1}j})^{\otimes n}\otimes\cdots\otimes A_k(e^{\frac{2\pi i}{ne+1}j})^{\otimes n}\right)(\psi^{\otimes n})\otimes|0\ldots 0\rangle  \\
  = \frac{1}{ne+1}\sum_{j=1}^{ne+1}e^{-\frac{2\pi i}{ne+1}ndj}\left(e^{\frac{2\pi i}{ne+1}ndj}\varphi^{\otimes n}+e^{\frac{2\pi i}{ne+1}(ndj+j)}(\ldots)+\cdots+e^{\frac{2\pi i}{ne+1}(ndj+nej)}\varphi_e^{\otimes n}\right)\otimes|0\ldots 0\rangle  \\
  = \varphi^{\otimes n}\otimes|0\ldots 0\rangle\slto\varphi^{\otimes n}
\end{multline}
By Proposition \ref{prop:basica} we have
\begin{equation}
\begin{split}
\omega(\psi,\varphi)
 & \le \omega(\psi,\psi^{\otimes n}\otimes GHZ_{ne+1})\omega(\psi^{\otimes n}\otimes GHZ_{ne+1},\varphi^{\otimes n})\omega(\varphi^{\otimes n},\varphi)  \\
 & \le \left(\omega(\psi,\psi^{\otimes n})+\omega(\psi,GHZ_{ne+1})\right)\frac{1}{n}  \\
 & = \left(n+\omega(\psi,GHZ)\log_2(ne+1)\right)\frac{1}{n}\to 1
\end{split}
\end{equation}
as $n\to\infty$.

If $\psi$ is separable across a bipartite cut $\cut{S}{\bar{S}}$ then $\varphi$ is also separable across this cut, because the set of $\cut{S}{\bar{S}}$-biseparable states is closed. But if $\psi=\psi_S\otimes\psi_{\bar{S}}$ degenerates to $\varphi=\varphi_S\otimes\varphi_{\bar{S}}$ then the two parts degenerate separately, so we can lift the condition that $\psi$ is globally entangled.
\end{proof}

With this powerful result at hand it is easy to prove that GHZ states can be transformed to W states by SLOCC asymptotically at rate $1$:
\begin{thm}\label{thm:GHZtoW}
Let $k\in\mathbb{N}$ and $W=|10\ldots 0\rangle+|010\ldots 0\rangle+\cdots+|0\ldots 01\rangle$ be the $k$-partite generalized W state. Then
\begin{equation}
\omega(GHZ,W)=1
\end{equation}
\end{thm}
\begin{proof}
The local rank of both states is $2$ across any bipartition, so $\omega(GHZ,W)\ge 1$. GHZ is SLOCC-equivalent to $|00\ldots 0\rangle-|11\ldots 1\rangle$ and
\begin{equation}\label{eq:Wappr2}
\left(\begin{array}{cc}
1 & -1 \\
\varepsilon & 0
\end{array}\right)\otimes\cdots\otimes
\left(\begin{array}{cc}
1 & -1 \\
\varepsilon & 0
\end{array}\right)(|00\ldots 0\rangle-|11\ldots 1\rangle)=\varepsilon W+O(\varepsilon^2)
\end{equation}
shows that the latter degenerates to W, therefore $\omega(GHZ,W)\le 1$. Note that equation \eqref{eq:Wappr2} is essentially the same as \eqref{eq:Wappr1} for $k=3$ and the error degree is $k-1$ in general.
\end{proof}

This result has also been obtained in \cite{Wrate1}, by showing that the tensor rank of $W^{\otimes n}$ is $O(n^{k-1}2^n)$ for fixed $k$ as $n\to\infty$. Our proof improves this bound to $(n(k-1)+1)2^n$. Note that the best lower bound found so far is $(k-1)2^n-k+2$ \cite{catalysis}.

\section{Symmetric states}\label{sec:symm}

In this section we generalize the result on the asymptotic rank of W states to certain symmetric states. In \cite{catalysis} it was shown that the tensor rank of the Dicke state $D_{m,n}$ that is the symmetrization of $|00\ldots 011\ldots 1\rangle$ with $m$ $0$-s and $n$ $1$-s is $\max\{n,m\}+1$. First we show that its border rank and its asymptotic rank are both $\min\{n,m\}+1$. Without loss of generality we can suppose $m\le n$. We can write
\begin{equation}
\frac{1}{m+1}\sum_{j=1}^{m+1}e^{-\frac{2\pi i}{m+1}mj}\left(|0\rangle+\varepsilon e^{\frac{2\pi i}{m+1}j}|1\rangle\right)\otimes\cdots\otimes\left(|0\rangle+\varepsilon e^{\frac{2\pi i}{m+1}j}|1\rangle\right)=\varepsilon^nD_{m,n}+O(\varepsilon^{2m+1})
\end{equation}
which implies $\brk D_{m,n}\le m+1$. For the lower bound consider the bipartite cut where the first $m$ and the last $n$ subsystems are the two parts. The rank across this cut is precisely $m+1$, so $\brk D_{m+n}\ge m+1$ and the same lower bound holds for the asymptotic rank. For the conversion rate this implies $\omega(GHZ,D_{m,n})=\log_2(\min\{m,n\}+1)$. Remarkably, if we keep the number of ``excitations'' $n$ fixed and let $m\to\infty$ the asymptotic conversion rate remains bounded, while the rank grows linearly.

In the following we investigate the conversion rates of more general symmetric states having similar extensions to more subsystems. We make the following definition:
\begin{defn}
Let $\lambda$ be an integer partition of $k_0$, i.e. $\lambda=(\lambda_1,\ldots,\lambda_d)$ with $\lambda_1\ge\lambda_2\ge\ldots\ge\lambda_d$ and $\lambda_1+\cdots+\lambda_d=k_0$. Let $k\ge k_0$ be an integer and define the following state:
\begin{equation}
D_{\lambda,k}:=\sum_{(i_1,\ldots,i_k)\in I_k}|i_1i_2\ldots i_k\rangle
\end{equation}
where $I_k\subseteq\{0,1,\ldots,d\}^k$ is the set of $k$-tuples containing the entry $i$ at exactly $\lambda_i$ positions for $1\le i\le d$, and $0$ at $k-k_0$ positions.
\end{defn}

In an earlier draft of this paper we proved that $\omega(GHZ,D_{\lambda,k})$ is bounded as $k\to\infty$ by finding an explicit upper bound on $\brk D_{\lambda,k}$ not depending on $k$, but we could not say how tight that bound is. As R. Duan explained to us, the proof in the appendix of \cite{Wrate1} can be formulated in the present framework and gives a much better bound -- in fact one that agrees with the exact value for large $k$. Here we outline how their result translates to a bound on the border rank.

First observe that
\begin{equation}
\frac{1}{\lambda_1!\cdots\lambda_d!}
\left(\frac{\partial}{\partial\varepsilon_1}\right)^{\lambda_1}\cdots\left(\frac{\partial}{\partial\varepsilon_d}\right)^{\lambda_d}\big(|0\rangle+\varepsilon_1|1\rangle+\ldots+\varepsilon_d|d\rangle\big)\otimes\cdots\otimes\big(|0\rangle+\varepsilon_1|1\rangle+\ldots+\varepsilon_d|d\rangle\big)\Big|_{\lambda_i=0}=D_{\lambda,k}
\end{equation}
where the number of tensor factors is $k$. The left hand side can be realized as the limit of finite differences using
\begin{equation}
\frac{\partial}{\partial x}f(x)\Big|_{x=0}=\lim_{h\to 0}\frac{1}{h^n}\sum_{i=0}^n(-1)^i\binom{n}{i}f\big((n-i)h\big)
\end{equation}
for each of the variables. Without the limit, this replacement results in a linear combination depending on $\varepsilon_1,\ldots,\varepsilon_d$ with $(\lambda_1+1)\cdots(\lambda_d+1)$ terms, each of which has rank $1$, therefore the rank of the sum is at most $(\lambda_1+1)\cdots(\lambda_d+1)$. Since this remains true no matter how small values we substitute for $\varepsilon_i$, we can conclude that $\brk D_{\lambda,k}\le(\lambda_1+1)\cdots(\lambda_d+1)$.

In general, for small $k$ we expect the asymptotic rank to become smaller than this upper bound, but using the bipartite rank ref. \cite{Wrate1} shows that the asymptotic rank and therefore also the border rank is at least $(\lambda_1+1)\cdots(\lambda_d+1)$ when $k\ge\lambda_1+\cdots+\lambda_k+\prod_{i=1}^d(\lambda_i+1)$. This means that for large $k$ the asymptotic rank stays constant and is equal to the above bound. Note that, in contrast, the tensor rank $\rk D_{\lambda,k}^{\otimes n}$ grows at least linearly in $k$ for fixed $\lambda$ and $n$, since $D_{\lambda,k}^{\otimes n}\slto D_{|\lambda|,k-|\lambda|}$.

It would be desirable to find the asymptotic rank of any Dicke state $D_{\lambda,k}$. Unfortunately, this appears to be difficult. Coppersmith and Winograd \cite{CW} conjecture that $\omega(GHZ_3,D_{(1,1,1),3})=1$, but proving this would imply that the exponent of matrix multiplication is $2$, a long-standing open problem in algebraic complexity theory.

\section{Transforming W states into GHZ states}\label{sec:WtoGHZ}

We turn to transformations in the opposite direction. It is well known that a W state cannot be transformed into a GHZ state, but it is easy to see that $k-1$ W states are enough to create a GHZ state: if $k-2$ of the parties perform the transformation $|0\rangle\mapsto|0\rangle,|1\rangle\mapsto 0$, then the remaining two end up sharing an EPR pair. Thus $k-1$ W states are enough to create EPR pairs between the first and each of the remaining $k-1$ parties, and then they can use teleportation to produce a GHZ state. In the following we will see that $\omega(W,GHZ)$ is in fact sublinear in the number of parties.

For tripartite systems an upper bound on the conversion rate $\omega(W,GHZ)$ has essentially been computed as a byproduct by Coppersmith and Winograd \cite{CW}. Later Strassen \cite{Strassen3} introduced a family of monotones for any number of parties which show that this upper bound is in fact optimal. In this section we are going to generalize the construction in \cite{CW} and use the monotones to prove optimality, thereby finding the exact values of $\omega(W,GHZ)$ for any number of subsystems.

First we summarize the relevant definitions and theorems by Strassen \cite{Strassen3}, formulated in terms of SLOCC conversion rates and for any number of parties. To this end we need to fix some more notations. The set
\begin{equation}
\Theta:=\left\{(\theta_1,\ldots,\theta_k)\in\mathbb{R}^{k}\big|\theta_1+\ldots+\theta_k=1, \forall j\in[k]:\theta_j\ge 0\right\}
\end{equation}
is the standard $k$-simplex. Given finite sets $I_1,\ldots,I_k$ and a probability measure on (a subset of) $I_1\times\cdots\times I_k$, its marginals are denoted by $P_1,\ldots,P_k$ and for $\theta\in\Theta$ we set
\begin{equation}
H_\theta(P)=\sum_{j=1}^k\theta_jH(P_j)
\end{equation}
with $H(P_j)=-\sum_{i\in I_j}P_j(i)\log_2 P_j(i)$.
For $\emptyset\neq\Psi\subseteq I_1\times\cdots\times I_k$ we set $H_\theta(\Psi)=\max_P H_\theta(P)$ where the maximization is over probability measures on $\Psi$, and extend this as $H_\theta(\emptyset)=-\infty$.

For $f\in \mathcal{H}_1\otimes\cdots\otimes \mathcal{H}_k$ is a tensor and a $k$-tuple of ordered bases $C=((u_{j,i})_{i=1}^{d_j})_{j=1}^k$ (one for each vector space), we can form the coordinate array $(f_{i_1,\ldots,i_k})_{i_j=1}^{d_j}\in\mathbb{C}^{d_1\times\cdots\times d_k}$ of $f$. We introduce the notation
\begin{equation}
\supp_Cf=\{(i_1,\ldots,i_k)|f_{i_1,\ldots,i_k}\neq 0\}\subseteq[d_1]\times\cdots\times[d_k]
\end{equation}
for the support of $f$ with respect to $C$. The set of such $k$-tuples of ordered bases will be denoted by $\mathcal{C}(\mathcal{H}_1,\ldots,\mathcal{H}_k)$ or simply $\mathcal{C}$ when the vector spaces are clear from the context.

For $f\in \mathcal{H}_1\otimes\cdots\otimes \mathcal{H}_k$ and $\theta\in\Theta$ as before, we set
\begin{equation}\label{eq:uppersf}
\rho^\theta(f)=\min_{C\in\mathcal{C}}H_\theta(\supp_Cf)\qquad\text{and}\qquad\zeta^\theta(f)=2^{\rho^\theta(f)}
\end{equation}
and call $\zeta^\theta$ the \emph{upper support functional}.

\begin{thm}[Strassen]
Let $\theta\in\Theta$. For the upper support functional $\zeta^\theta$ the followings hold:
\begin{enumerate}
\item $\zeta^\theta(GHZ_r)=r$ for $r\in\mathbb{N}$
\item $\zeta^\theta(f\oplus g)=\zeta^\theta(f)+\zeta^\theta(g)$ for all tensors $f$,$g$
\item $\zeta^\theta(f\otimes g)\le\zeta^\theta(f)\zeta^\theta(g)$ for all tensors $f$,$g$
\item $f\slto g\implies\zeta^\theta(f)\ge\zeta^\theta(g)$
\item $\zeta^\theta(f)\in[0,d_1^{\theta_1}d_2^{\theta_2}\ldots d_k^{\theta_k}]$ for $f\in \mathcal{H}_1\otimes\cdots\otimes \mathcal{H}_k$ with $d_j=\dim \mathcal{H}_j$ ($j\in[k]$).
\end{enumerate}
\end{thm}
The proof can be found in \cite[sec. 2]{Strassen3} (explicitly only for $k=3$, but the generalization is straightforward).

For $f\in \mathcal{H}_1\otimes\cdots\otimes \mathcal{H}_k$ and $\theta\in\Theta$ as before, we set
\begin{equation}\label{eq:lowersf}
\rho_\theta(f)=\max_{C\in\mathcal{C}}H_\theta(\max\supp_Cf)\qquad\text{and}\qquad\zeta_\theta(f)=2^{\rho_\theta(f)}
\end{equation}
where $\max\supp_Cf$ denotes the set of maximal points in the support with respect to the product partial order, and call $\zeta_\theta$ the \emph{lower support functional}.

\begin{thm}[Strassen]
Let $\theta\in\Theta$. For the lower support functional $\zeta_\theta$ the followings hold:
\begin{enumerate}
\item $\zeta_\theta(GHZ_r)=r$ for $r\in\mathbb{N}$
\item $\zeta_\theta(f\oplus g)\ge\zeta_\theta(f)+\zeta_\theta(g)$ for all tensors $f$,$g$
\item $\zeta_\theta(f\otimes g)\ge\zeta_\theta(f)\zeta_\theta(g)$ for all tensors $f$,$g$
\item $f\slto g\implies\zeta_\theta(f)\ge\zeta_\theta(g)$
\item $\zeta_\theta(f)\in[0,d_1^{\theta_1}d_2^{\theta_2}\ldots d_k^{\theta_k}]$ for $f\in \mathcal{H}_1\otimes\cdots\otimes \mathcal{H}_k$ with $d_j=\dim \mathcal{H}_j$ ($j\in[k]$).
\end{enumerate}
\end{thm}
The proof can be found in \cite[sec. 3]{Strassen3}.

The two functionals are related as follows:
\begin{thm}[Strassen]
For any $\theta\in\Theta$ and any pair of tensors $f,g$ the inequality $\zeta^\theta(f\otimes g)\ge\zeta^\theta(f)\zeta_\theta(g)$ holds, so in particular, $\zeta_\theta(g)\le\zeta^\theta(g)$.
\end{thm}
The proof can be found in \cite[sec. 4]{Strassen3}.

Note that $\rho_\theta(f)$ as a function of $\theta$ is a maximum of affine functions, and hence convex. $\zeta_\theta(f)$ is the composition of $\rho_\theta(f)$ with the increasing convex function $x\mapsto 2^x$ and hence also convex.

We say that a tensor $f$ is \emph{$\theta$-robust} \cite{Strassen3} when $\zeta_\theta(f)=\zeta^\theta(f)$ and \emph{robust} when it is $\theta$-robust for all $\theta\in\Theta$. By the properties of $\zeta^\theta$ and $\zeta_\theta$ we see that $f\in \mathcal{H}_1\otimes\cdots\otimes \mathcal{H}_k$ is $\theta$-robust iff there are $C,C'\in\mathcal{C}(\mathcal{H}_1,\ldots,\mathcal{H}_k)$ and a probability measure $P$ on $\max\supp_Cf$ such that $H_\theta(P)\ge H_\theta(\supp_{C'}f)$.

Let us call $f$ \emph{oblique} when there exists $C\in\mathcal{C}(\mathcal{H}_1,\ldots,\mathcal{H}_k)$ such that $\supp_Cf$ is an antichain in the product partial order (i.e. no two elements are comparable). An oblique tensor is also robust. The sets of $\theta$-robust, robust and oblique tensors are closed under direct sums and tensor products. This implies that for robust tensors the upper and lower support functionals are multiplicative, because one of them is always submultiplicative while the other is always supermultiplicative, and the two agree on tensor powers of a robust tensor. Because of this multiplicativity one can use the upper and lower support functionals to bound the SLOCC conversion rates for robust tensors as follows.

Let $\psi$ and $\varphi$ be robust $k$-partite states. By definition, for each $n$ one can transform $n\omega_n(\psi,\varphi)$ copies of $\psi$ to $n$ copies of $\varphi$, and therefore, for each $\theta\in\Theta$ we have
\begin{equation}
\zeta_\theta(\psi)^{n\omega_n(\psi,\varphi)}=\zeta_\theta(\psi^{\otimes n\omega_n(\psi,\varphi)})\ge \zeta_\theta(\varphi^{\otimes n})=\zeta_\theta(\varphi)^n
\end{equation}
After taking logarithms and letting $n\to\infty$ we get $\omega(\psi,\varphi)\rho_\theta(\psi)\ge\rho_\theta(\varphi)$. Since this holds for any $\theta\in\Theta$ we can also write
\begin{equation}
\omega(\psi,\varphi)\ge\max_{\theta\in\Theta}\frac{\rho_\theta(\varphi)}{\rho_\theta(\psi)}
\end{equation}

The construction in \cite{CW} makes use of Salem-Spencer sets \cite{SalemSpencer}. These are ``large'' sets of numbers containing no nontrivial three-term arithmetic progressions. We introduce the following generalization of Salem-Spencer sets:
\begin{defn}
A subset $S\subset\mathbb{N}$ is \emph{$m$-average-free} if $A_1,A_2,\ldots,A_m,B\in S$ and $A_1+\cdots+A_m=mB$ implies $A_1=A_2=\ldots=A_m=B$.

The maximum size of an $m$-average-free subset $S\subseteq[N]$ will be denoted by $\nu_m(N)$.
\end{defn}
Note that $2$-average-free sets are precisely sets without three-term arithmetic progressions. The next lemma extends the result of Salem and Spencer:
\begin{lem}\label{thm:SalemSpencer}
For any fixed value of $m$ we have $\nu_m(N)=N^{1-o(1)}$ as $N\to\infty$.
\end{lem}
The proof is very similar to that of \cite{SalemSpencer}, therefore we leave it to the appendix.

Now we have everything at hand to find the conversion rate from W to GHZ states.
\begin{thm}\label{thm:WtoGHZ}
Let $k\ge 2$ and consider the $k$-partite states W and GHZ. We have
\begin{equation}
\omega(W,GHZ)=\frac{1}{h\left(\frac{1}{k}\right)}
\end{equation}
with $h(p)=-p\log_2 p-(1-p)\log_2(1-p)$ the binary entropy function.
\end{thm}
\begin{proof}
First we prove that the right hand side is a lower bound
The support of the W states in the computational basis is $supp_CW=\{(1,0,\ldots,0),(0,1,0,\ldots,0),\ldots,(0,\ldots,0,1)\}$, which is clearly an antichain. Hence W is robust and $\zeta^\theta(W)=\zeta_\theta(W)$ and $\zeta_\theta(W^{\otimes n})=\zeta_\theta(W)^n$. Comparing to GHZ states is especially easy since $\zeta_\theta(GHZ)=2$ independently of $\theta$. As W is invariant under permutations of the subsystems, $\zeta_\theta$ is also invariant under permutations of the coordinates of $\theta$. Together with convexity of $\rho_\theta$ we conclude that
\begin{equation}
\omega(W,GHZ)\ge\max_{\theta\in\Theta}\frac{\rho_\theta(GHZ)}{\rho_\theta(W)}=\max_{\theta\in\Theta}\frac{1}{\rho_\theta(W)}=\frac{1}{\rho_{(\frac{1}{k},\ldots,\frac{1}{k})}(W)}=\frac{1}{\rho^{(\frac{1}{k},\ldots,\frac{1}{k})}(W)}
\end{equation}
To compute $\rho^{\frac{1}{k},\ldots,\frac{1}{k}}(W)$ we observe that $H_{(\frac{1}{k},\ldots,\frac{1}{k})}(P)$ is permutation-invariant and concave (a linear combination of concave functions composed with taking marginals, which are affine maps). This means that the optimal distribution is uniform, and the marginal distributions are $(\frac{1}{k},1-\frac{1}{k})$. Thus we have proved
\begin{equation}
\omega(W,GHZ)\ge\frac{1}{h\left(\frac{1}{k}\right)}
\end{equation}

Next we prove the inequality in the other direction using a generalization of the proof by Coppersmith and Winograd \cite{CW} to $k>3$ parties.
Fix an $\varepsilon>0$. According to prop. \ref{thm:SalemSpencer} there is a constant $C_\varepsilon>0$ such that for all $N$ $\nu_{k-1}(N)\ge C_\varepsilon N^{1-\varepsilon}$. Choose such a $C_\varepsilon$.

For some $n\in\mathbb{N}$ consider the $kn$th tensor power of W, and for simplicity, reverse the roles of $0$ and $1$, so
\begin{equation}
W=|011\ldots 1\rangle+|101\ldots 1\rangle+\ldots+|1\ldots 10\rangle
\end{equation}
The state $W^{\otimes kn}$ can be written a sum of $k^{kn}$ terms. To each such term we associate $k$ vectors $I_1,\ldots,I_k\in\{0,1\}^{kn}$ in the following way. $I_i$ corresponds to the $i$th site, $(I_i)_j$ to the $j$th tensor factor and the value of $(I_i)_j=0$ iff there is $|0\rangle$ at that place and $1$ otherwise. Call $w(I)=\text{number of $1$s in $I$}$ the weight of such a vector.

At each site $i$ apply the transformation which sends the terms with weight not equal to $(k-1)n$ to the zero vector and leaves the rest unaltered. For some $\alpha>0$ to be chosen later take $M=(k-1)\lfloor\alpha^n\rfloor+1$ and let $\mathcal{B}$ be a $k-1$-average-free set with $\max\mathcal{B}<\frac{M}{k-1}$ and $|\mathcal{B}|\ge C_\varepsilon \left(\frac{M}{k-1}-1\right)^{1-\varepsilon}$.

Select independent random integers $0\le v_j<M$ ($j=0,1,\ldots,kn$) and $0\le u_i<M$ ($i=1,\ldots,k-1$) uniformly. For the vectors $I_i$ compute a hash as follows:
\begin{equation}
\begin{split}
b_i(I_i) & =u_i+\sum_{j=1}^{kn}(I_i)_jv_j\qquad\text{for $1\le i<k$ and}  \\
b_k(I_k) & =\frac{1}{k-1}(u_1+\ldots+u_{k-1}+\sum_{j=1}^{kn}(k-1-(I_k)_j)v_j)
\end{split}
\end{equation}
where the operations are to be understood mod $M$ ($k-1$ is invertible mod $M$ because $\gcd(k-1,M)=1$).

By construction, for any term we have $b_1(I_1)+\ldots+b_{k-1}(I_{k-1})-(k-1)b_k(I_k)\equiv 0\pmod{M}$. Now apply local transformations which send to zero the terms for which $b_i(I_i)\notin\mathcal{B}$. The sum of $k-1$ such numbers is the same as their sum mod $M$, so for the remaining terms we have $b_1(I_1)+\ldots+b_{k-1}(I_{k-1})=(k-1)b_k(I_k)$, and therefore by the $k-1$-average-free property $b_1(I_1)=b_2(I_2)=\ldots=b_k(I_k)$.

What we have achieved so far is that the remaining state decomposes into blocks labelled by the $b$ values, so it is in a sense closer to the GHZ form. However, inside each block the terms can follow different patterns. Our goal is to erase some of them with local transformations in such a way that the remaining terms give a GHZ state, i.e. each term is uniquely determined by any of its $k$ tensor factors. To this end, for each $b\in\mathcal{B}$ we make a list
\begin{multline}
L_b=\{(I_1,\ldots,I_k)|w(I_1)=w(I_2)=\ldots=w(I_k)=(k-1)n\text{, } \\ I_1+\ldots+I_k=(k-1,k-1,\ldots,k-1)\text{, }b_1(I_1)=b_2(I_2)=\ldots=b_k(I_k)=b\}
\end{multline}
and consider its elements in some arbitrarily fixed order. For each $k$-tuple in the list if it shares a vector $I_i$ with one $k$-tuple occuring earlier in the list, we set to zero the terms corresponding to one of the other vectors $I_{i'}$ and thus eliminate this $k$-tuple.

We need some control on the number of terms erased, so that we get a large enough GHZ state. The expected number of entries in the list before eliminating is
\begin{equation}
\mean|L_b|=\binom{kn}{n,n,\ldots,n}M^{-(k-1)}
\end{equation}
because $\binom{kn}{n,n,\ldots,n}$ is the number of compatible vectors, $\Pr(B_i(I_i)=b)=M^{-1}$ and these events for $1\le i\le k-1$ are independent.

For any $i=1,\ldots,k$ the expected number of pairs of $k$-tuples with common $i$th vector is
\begin{equation}
\frac{1}{2}\binom{kn}{n,n,\ldots,n}\left(\binom{(k-1)n}{n,n,\ldots,n}-1\right)M^{-(k-1)}M^{-(k-2)}
\end{equation}
Here the $\frac{1}{2}$ is needed to count unordered pairs, the first binomial coefficient is the number of ways one can select a $k$-tuple of vectors $(I_1,\ldots,I_k)$, the next factor is the number of $k$-tuples $(I'_1,\ldots,I'_k)$ with $I_i=I'_i$ and $M^{-(k-1)}M^{-(k-2)}$ is the probability for each of these pairs to end up in the block $b$.

Suppose we are about to eliminate a vector $I_i$ because a pair of $k$-tuples shares the vector $I_{i'}$ at their $i'$th position. If $L$ of the remaining $k$-tuples share this $I_i$ then setting the terms with $I_i$ to zero eliminates these $L$ $k$-tuples while eliminating at least $\binom{L}{2}+1$ colliding pairs (those sharing $I_i$ and at least one pair sharing $I_{i'}$). Since $\binom{L}{2}+1\ge L$ during this procedure we eliminate at least as many pairs of colliding $k$-tuples as $k$-tuples. Using $M=(k-1)\lfloor\alpha^n\rfloor+1=\alpha^{n+o(n)}$ the expected number of remaining $k$-tuples on one list is at least
\begin{multline}
\binom{kn}{n,\ldots,n}M^{1-k}-\frac{k}{2}\binom{kn}{n,\ldots,n}\left(\binom{(k-1)n}{n,\ldots,n}-1\right)M^{3-2k}  \\
 = k^{kn+o(n)}\left(\alpha^{(n+o(n))(1-k)}-(k-1)^{(k-1)n+o(n)}\alpha^{(3-2k)n+o(n)}\right)  \\
 = k^{kn+o(n)}\left((\alpha^{1-k+o(1)})^n-((k-1)^{k-1+o(1)}\alpha^{3-2k+o(1)})^n\right)
\end{multline}
As long as $(k-1)^{k-1}\alpha^{3-2k}<\alpha^{1-k}$ this grows like $(k^k\alpha^{1-k+o(1)})^n$. Take $\alpha=(k-1)^{\frac{k-1}{k-2}}+\varepsilon$ to ensure this.

The number of lists is $|\mathcal{B}|=\alpha^{n+o(n)}$, so the expected number of remaining $k$-tuples is at least
\begin{equation}
\left(k^k((k-1)^{\frac{k-1}{k-2}}+\varepsilon)^{2-k+o(1)}\right)^n
\end{equation}
It follows that there is a constant $C'_\varepsilon$ and for each $n$ particular values of the random variables $u_i,v_j$ such that the number of $k$-tuples remaining after the elimination is at least
\begin{equation}
N_n\ge C'_\varepsilon \left(k^k((k-1)^{\frac{k-1}{k-2}}+\varepsilon)^{2-k-\varepsilon}\right)^n
\end{equation}
The resulting state is $GHZ_{N_n}$, and therefore
\begin{multline}
\omega(W,GHZ)\le\omega(W,W^{\otimes kn})\omega(W^{\otimes kn},GHZ_{N_n})\omega(GHZ_{N_n},GHZ)  \\
\le\frac{kn}{\log_2 N_n}=\frac{kn}{\log_2C'_\varepsilon+n\log_2\left(k^k((k-1)^{\frac{k-1}{k-2}}+\varepsilon)^{2-k-\varepsilon}\right)}
\end{multline}
Now let $n\to\infty$ and then $\varepsilon\to 0$.

\end{proof}

Note that as $k\to\infty$ the sequence $\omega(W,GHZ)$ is asymptotically equal to $\frac{k}{\log_2k}$. In contrast, the trivial upper bound was linear in $k$.

\section{Conclusion}

We used the concept of border rank from algebraic complexity theory to bound asymptotic SLOCC conversion rates from GHZ to symmetric states. In the case of multiqubit W and Dicke states this bound gives the exact values. In some cases this results in a dramatic improvement to the bound obtained using tensor rank of an arbitrary fixed number of copies.

In the other direction, we prove that $n$ copies of the $k$-qubit W state can be transformed into $h(\frac{1}{k})n+o(n)$ GHZ states. This follows from a generalization of a construction by Coppersmith and Winograd, optimal as shown using the monotones introduced by Strassen. These monotones are functions $\zeta^\theta$, $\zeta_\theta$ of the states (see eqs. \eqref{eq:uppersf} and \eqref{eq:lowersf}), not increasing under SLOCC, and are not restricted to the asymptotic regime, although they become particularly powerful in this setting.

Observe that a $k$ qubit $W$ state is the symmetrization of $|100\ldots 0\rangle$ and the empirical distribution of this sequence of bits is $(\frac{1}{k},1-\frac{1}{k})$, the entropy of which gives the rate at which GHZ states can be extracted from W states. It seems likely that our proof can be extended to show that
\begin{equation}
\omega(D_\lambda)=\frac{1}{H\left(\frac{\lambda_1}{|\lambda|},\frac{\lambda_2}{|\lambda|},\ldots,\frac{\lambda_k}{|\lambda|}\right)}
\end{equation}
with $H(p_1,\ldots,p_k)=-\sum_i p_i\log_2p_i$. Simple calculation shows that here the denominator is equal to $\rho_\theta(\omega(D_\lambda))$, so the right hand side is a lower bound on the left hand side.

It would be desirable, but much more difficult to find the conversion rates from GHZ states to an arbitrary $D_{\lambda,k}$, or at least the limit for large $k$. We leave these problems as open questions for further study.

\section{Acknowledgements}

We would like to thank Harry Buhrman, Peter B\"urgisser and Runyao Duan for helpful discussions. We acknowledge a Sapere Aude grant of the Danish Council for Independent Research, an ERC Starting Grant, the CHIST-ERA project ``CQC'', an SNSF Professorship, the Swiss NCCR ``QSIT'' and the Swiss SBFI in relation to COST action MP1006.

\appendix

\section{Generalized Salem-Spencer sets}

\begin{proof}[Proof of Lemma \ref{thm:SalemSpencer}]
Let $d,n\in\mathbb{N}$ with $d|n$ and consider the set
\begin{multline}
S_m(d,n)=\bigg\{a_1+(md-(m-1))a_2+\ldots+(md-(m-1))^{n-1}a_n\bigg|\text{$(a_1,\ldots,a_n)$ contains}  \\  \text{$\frac{n}{d}$ $0$s, $1$s, \ldots,$d-1$s}\bigg\}
\end{multline}
Clearly
\begin{equation}
\left|S_m(d,n)\right|=\frac{n!}{\left(\frac{n}{d}!\right)^{d}}\qquad\text{and}\qquad \max S_m(d,n)\le(md-(m-1))^n
\end{equation}
This set is $m$-average-free, because summing $m$ elements cannot produce a carry digit in base $(md-(m-1))$ and therefore a $0$ digit in the sum forces all the terms to have $0$ at that same position, and by induction this is true for all digits.

For a given $N\in\mathbb{N}$ take the unique $d$ such that $(md-(m-1))^d\le N<(m(d+1)-(m-1))^{d+1}$. With this choice $S_m(d,d)\subseteq[N]$ and therefore
\begin{equation}
\begin{split}
1
 & \ge\frac{\log\nu_m(N)}{\log N}\ge\frac{\log|S_m(d,d)|}{\log N}\ge\frac{\log d!}{\log(m(d+1)-(m-1))^{d+1}}  \\
 & =\frac{\log d!}{(d+1)\log(md-1)}=\frac{d\log d-d+\log\sqrt{2\pi d}+O(\frac{1}{d})}{(d+1)\log(md+1)}\to 1
\end{split}
\end{equation}
as $N\to\infty$ and therefore also $d\to\infty$.
\end{proof}

\end{document}